\documentclass[10pt,aps,pra,amsfonts,amssymb,amsmath,
floatfix,nofootinbib, groupedaddress,superscriptaddress,citesort,twocolumn
]{revtex4}%

\usepackage{graphicx}
\usepackage{subfigure}       
\usepackage{caption}      
\usepackage{xcolor}       
\usepackage{amsthm}
\usepackage{mathrsfs}
\usepackage{amsfonts}
\usepackage{amstext}
\usepackage{bm}
\usepackage{bbm}
\usepackage{hyperref}     
\usepackage{caption} 
\makeatletter
\let\old@makecaption\@makecaption
\renewcommand{\@makecaption}[2]{\old@makecaption{#1}{\ignorespaces #2}}
\makeatother

\hypersetup{
	hyperindex=true,
	breaklinks=true,
	hidelinks,
	colorlinks=false}

\newtheorem{theorem}{Theorem}

\newtheorem{proposition}[theorem]{Proposition}
\newtheorem{example}[theorem]{Example}
\newtheorem{definition}[theorem]{Definition}

\begin{document}
	
	\preprint{APS/123-QED}
	
	\title{ Several kinds of  Gaussian quantum  channels  related to Einstein-Podolsky-Rosen steering}

	\author{Ruifen Ma}\thanks{Corresponding author}
		\email{ruifenma@tyust.edu.cn}
	\affiliation{Department of Mathematics, Taiyuan University of
		Science and Technology, Taiyuan 030024, P. R. China}

	\author{Yanjing Sun}
	\affiliation{Department of Mathematics, Taiyuan University of
		Science and Technology, Taiyuan 030024, P. R. China}

	\author{Xiaofei Qi}\thanks{Corresponding author}
\email{xiaofeiqisxu@aliyun.com}
\affiliation{School of Mathematics and Statistics, Shanxi University,
	Taiyuan 030006, P. R. China}\affiliation{Key Laboratory of Complex Systems and Data Science of Ministry of Education,
	Shanxi University, Taiyuan  030006,  Shanxi, China}

\begin{abstract}
EPR steering is a crucial quantum resource that lies intermediate between entanglement and Bell nonlocality. Gaussian channels, meanwhile, play a foundational role in diverse quantum protocols, secure communication, and related fields. In this paper, we focus on several classes of Gaussian channels associated with EPR steering: Gaussian steering-annihilating channels, Gaussian steering-breaking channels, Gaussian unsteerable channels, and maximal Gaussian unsteerable channels. We give the concepts of these channels, derive the necessary and sufficient conditions for a Gaussian channel to belong to each class, and explore the intrinsic relationships among them. Additionally, since quantifying the steering capability of Gaussian channels in continuous-variable systems requires an understanding of the structure of free superchannels, we also provide a detailed characterization of Gaussian unsteerable superchannels and maximal Gaussian unsteerable superchannels.
\end{abstract}

\maketitle
\section{Introduction}

EPR steering is a fundamental and  important resource for quantum information science.
In 1935, Einstein, Podolsky and Rosen (EPR) first discovered the
anomalous phenomenon of quantum states in  multipartite quantum systems, which
is contrary to the classical mechanics \cite{EPR}. In order to
capture the essence of the EPR paradox, the notion of EPR steering
was  introduced by Schr$\ddot{o}$dinger  in \cite{S}, which is  a quantum correlation between entanglement and Bell
nonlocality. It has been shown that EPR steering plays a fundamental role in
various quantum protocols, secure communication and other fields
\cite{CGP,D1,hq}.

Gaussian states are a special class of quantum states in continuous-variable (CV) quantum systems, playing a pivotal role in quantum optics and quantum information theory \cite{SLP,WPG1,S1}. Over the past few years, the EPR steering criteria and measures for Gaussian states have garnered considerable attention from researchers (see \cite{JiL,KSA,KLRA,ZYL,TGRF,MKIB,BJC,YHQ} and the references therein).
Notably, quantum information processing inevitably involves quantum channels. As a distinctive class of quantum channels, Gaussian quantum channels not only furnish a core theoretical framework for elucidating the inherent physical limitations of quantum communication and quantum computing, but also directly underpin the translation of quantum technologies from theoretical concepts to practical applications \cite{BSS,MDG,CEGH,CD,NGA,LD,PMGH}. They thus hold irreplaceable significance in quantum systems, especially in optical quantum systems.

 In quantum information processing, the storage and transmission of quantum states are crucial but inevitably affected by the environment. Thus, researching which types of environmental influences should be avoided and which are acceptable is of critical importance. In other words, analyzing the structure of quantum channels is essential. One promising research direction focuses on the dynamics of quantum resources, such as quantum entanglement, EPR steering and  non-locality, under local quantum channels.
Horodecki, Shor and Ruskai \cite{HSR} studied
entanglement-breaking channels in bipartite discrete-variable quantum systems which completely disentangle the subsystem they act on from the rest of the system, and  proved that the Holevo capacity of entanglement-breaking channels is additive. This work provides theoretical support for optimizing classical information transmission over noisy channels in quantum communication.
Holevo \cite{HEB} discussed the condition any mode Gaussian quantum channels becoming entanglement-breaking, and  using this condition to study several additivity conjectures of capacity for one mode Gaussian channels.
The results in  \cite{KGEB} indicate that  Gaussian entanglement-breaking channels and entanglement-annihilating channels serve as essential benchmarks for the infeasibility assessment of entanglement preservation, and act as core tools for both theoretical analysis and practical guidance in this field. Based on Gaussian entanglement-breaking channels, passive eavesdropping-immune CV quantum key distribution  protocols can be constructed.
Pereg \cite{P11} investigated the communication performance of entanglement-breaking channels with unreliable entanglement assistance, revealing the mechanism by which initial entanglement enhances the signal-to-noise ratio through classical correlations.  Hence the study on (Gaussian) entanglement-breaking channels is of great practical significance.

Note that  EPR steering is another important quantum  resource. Inspired by the above works,
the primary objective of this paper is to investigate some types of special Gaussian quantum channels with respect to EPR steering: Gaussian steering-breaking channels that locally disrupt steering; Gaussian steering-annihilating channels that completely eliminate steering;  and maximal Gaussian unsteerable channels that map Gaussian unsteerable states into Gaussian unsteerable states, so as to exploit the ability of Gaussian channels to create or destroy steering and  lay a solid theoretical foundation for the resource theory of steering for Gaussian channels.\\

Regarding the correlation measures and resource theories of quantum channels, substantial research efforts have been devoted. Bäuml et al. \cite{BDW} proposed several entanglement measures tailored for bipartite quantum channels. Mani \cite{MK} introduced the concepts of cohering and decohering power of quantum channels, along with corresponding quantification methods. Xu \cite{X19} established a coherence resource theory for channels in finite-dimensional systems, while the authors of \cite{YX1} developed a general operational resource theory framework for quantum channels in such systems. For Gaussian channel resource theories, Xu \cite{X11} constructed a coherence resource theory specific to Gaussian channels and proposed a coherence measure for them, based on the relative entropy coherence measure for Gaussian states.
Recall that a quantum channel
 resource theory is defined as a tuple $(\mathcal F,\mathcal O,\mathcal R)$, where $\mathcal F$ is the set of free
 channels  that do not have any
 resource, $\mathcal O$ is the set of  free superoperations  which transform free channels into free channels; and  $\mathcal R$  is the set of  channel resource measures which map quantum channels into  nonnegative real numbers satisfying the following two fundamental conditions:

$( f_1)$ non-negativity: $\mathcal R\left( \phi  \right) \ge 0$ for all $\phi  \in \mathcal C(H)$ (the set of all quantum channels on a separable complex Hilbert space $H$),  and
 $\mathcal R\left( \phi  \right)=0$ for any $\phi \in {\mathcal F}$;

$ (f_2)$ monotonicity: $\mathcal R\left( {\Psi\left( \phi
 	\right)} \right) \le \mathcal R\left( \phi \right)$ holds for all $\phi  \in
 {\mathcal F}$ and all $\Psi\in
 {\mathcal O}$. \\
To lay the foundation for the future development of a steering resource theory for Gaussian channels, another core objective of this paper is to investigate the structure of free superchannels, specifically Gaussian unsteerable superchannels and maximal Gaussian unsteerable superchannels.

This paper is structured as follows. In Section II, we review fundamental concepts related to continuous-variable (CV) systems, including Gaussian states, Gaussian channels, Gaussian unsteerable channels, and Gaussian quantum steering. In Section III, we formally define Gaussian steering-annihilating channels and Gaussian steering-breaking channels, analyze their structural properties, derive the necessary and sufficient conditions for a Gaussian channel to be classified as either type, and explore the relationships between these  channels. Section IV is dedicated to characterizing Gaussian unsteerable superchannels and maximal Gaussian unsteerable superchannels. Section V presents a concise summary of the work.

\section{Preliminaries}

In this section, we  briefly recall some notions and notations about  Gaussian states and  Gaussian quantum channels.

\subsection{ Gaussian states }

 Consider an $N$-mode CV system with state space $H=H_1\otimes H_2\otimes \cdots\otimes H_N$, where each $H_k \ (1\le k\le N)$ is an infinite-dimensional separable complex Hilbert space.
 Denote by $\mathcal{S}(H)$ the set of all quantum states (that is, positive bounded linear operators with trace 1) on $H$.

 For any state $\rho\in\mathcal{S}(H)$, its characteristic function $\chi_{\rho}$ is
 defined as
$$\chi_{\rho}(z)={\rm tr}(\rho W(z)),$$ where
$z=(x_{1}, y_{1}, \cdots, x_{N}, y_{N})^{\rm T}\in{\mathbb R}^{2N}$,
$W(z)=\exp(i{R}^{\rm T}z)$ is the Weyl displacement operator,
${R}=(R_1,R_2,\cdots,R_{2N})=(\hat{Q}_1,\hat{P}_1,\cdots,\hat{Q}_N,\hat{P}_N)$,
$\hat{Q}_k=(\hat{a}_k+{\hat{a}_k}^\dag)/\sqrt{2}$
and
$\hat{P_k}=-i({\hat{a}_k}-{\hat{a}_k}^\dag)/\sqrt{2}$
($k=1,2,\cdots,N$) are respectively the position and momentum operators. Here,
$\hat{a}_k^\dag$ and
$\hat{a}_k$ are the creation and annihilation operators in the $k$th
mode satisfying the Canonical Commutation Relation (CCR):
$$[\hat{a}_k,\hat{a}_l^\dag]=\delta_{kl}I\ {\rm and}
\ [\hat{a}_k^\dag,\hat{a}_l^\dag]=[\hat{a}_k,\hat{a}_l]=0,\ k,l=1,2,\cdots,N.$$ Particularly,  $\rho$ is called a Gaussian state
if $\chi_{\rho}(z)$ is of the form
\begin{eqnarray*}\label{N1}
	\chi_{\rho}(z)=\exp[-\frac{1}{4}z^{\rm T}\Gamma z+i{\mathbf
		d}^{\rm T}z],
\end{eqnarray*}
where  $$
\begin{array}{rl}{\mathbf d}=&(\langle\hat R_1 \rangle, \langle\hat R_2
\rangle, \ldots ,\langle\hat R_{2N} \rangle)^{\rm T}\\
=&({\rm tr}(\rho
R_1), {\rm tr}(\rho R_2), \ldots, {\rm tr}(\rho R_{2N}))^{\rm
	T}\in{\mathbb R}^{2N}\end{array}$$ is called the mean or the displacement
vector of $\rho$ and $\Gamma=(\gamma_{kl})\in \mathcal M_{2N}(\mathbb R)$ is called
the covariance matrix (CM) of $\rho$ defined by $\gamma_{kl}={\rm
	tr}[\rho
(\Delta\hat{R}_k\Delta\hat{R}_l+\Delta\hat{R}_l\Delta\hat{R}_k)]$
with $\Delta\hat{R}_k=\hat{R}_k-\langle\hat{R}_k\rangle$
\cite{SLP}. Here, $\mathcal M_d(\mathbb R)$ stands for the algebra of all
$d\times d$ matrices over the real field $\mathbb R$.
So,  any Gaussian state $\rho$ with CM $\Gamma$ and displacement
vector ${\mathbf d}$ will sometimes be represented as $\rho(\Gamma, {\mathbf	d})$.
Note that  $\Gamma$ is real symmetric and satisfies the
condition
\begin{eqnarray*}\label{relation1}
\Gamma +i\Omega_N\geq 0,
\end{eqnarray*}
where
$$ \Omega_N=\underbrace{\Omega\oplus\cdots\oplus\Omega}_{N} \ \ {\rm with}\ \ \Omega=\begin{pmatrix}0&1\\-1&0\end{pmatrix}.$$

 Now, divide the $N$-mode CV system into $m$-mode CV subsystem A and $n$-mode CV subsystem B,
with state space $H=H_A\otimes H_B$ and $N=m+n$. Assume that
$\rho$ is any $(m+n)$-mode bipartite Gaussian state. Then its CM $\Gamma_{\rho}$
can be written as
\begin{eqnarray}\label{N2}
	\Gamma_{\rho}=\left(\begin{array}{cc}A& C\\
		C^{\rm T} & B\end{array}\right),
\end{eqnarray}
where $A \in \mathcal M_{2m}({\mathbb R}),\  B
\in \mathcal M_{2n}({\mathbb R}), C\in \mathcal M_{2m\times 2n}({\mathbb R})$.
Particularly, if $n=m=1$, then  $\Gamma$ has the following standard form:
 \begin{equation}\label{standard}
	\Gamma = \left( {\begin{array}{*{20}{c}}
		a&0&c&0\\
		0&a&0&d\\
		c&0&b&0\\
		0&d&0&b
\end{array}} \right),
\end{equation}
where $a, b \ge 1, \ ab - c^2 \ge 1, \ ab - d^2 \ge 1$. For more details about Gaussian states, see \cite{WPG1,S1}.

\subsection{EPR steering}

 In a  	bipartite EPR steering scenario, Alice and Bob share  a bipartite state $\rho_{AB}\in{\mathcal S}(H_A\otimes H_B)$, and Alice performs
positive-operator-valued  measurements (POVMs)	on her subsystem to steer $\rho_{AB}$ on Bob's side.  If Alice performs a set of POVMs  ${\mathcal {MA}} = \{M_{a|x}\}_{a,x}$ (that is, $M_{a|x} \geq0$ with $\sum_a
M_{a|x} = I$ for each $x$), then the assemblage of sub-normalized ``conditional  states" of the
subsystem state $\rho_B={\rm Tr}_A(\rho_{AB})$ is $\{\rho^B_{a|x}\}_{a,x}$, where
$$\rho^B_{a|x}={\rm Tr}_A ((M^A_{a|x}\otimes I_B){\rho}_{AB}).$$
$\rho_{AB}$  is said to be unsteerable from A to B if every assemblage
$\{\rho^B_{a|x}\}_{a,x}$ on Bob's side can be explained by a  local hidden state (LHS) model as follows
$$\rho^B_{a|x}=\sum_\lambda p_\lambda p(a|x,\lambda)\sigma_\lambda,$$
where $\lambda$ is a hidden variable,
$p_\lambda$ is a distribution in $\lambda$,  $p(a|x,\lambda)$ are local
``response functions" of Alice, and  $\sigma_\lambda$
are ``hidden states" of Bob. Otherwise, $\rho_{AB}$ is called steerable (from A
to B) \cite{WJ}. 
Symmetrically,
we can define the steerability of $\rho_{AB}$ (from B to A).

 In CV systems, Gaussian POVM (GPOVM) plays an important role.
Recall that an
$N$-mode GPOVM
$\Pi=\{\Pi(\alpha)\}$  is defined  as
$$\Pi(\alpha)=\frac{1}{\pi^N}D(\alpha)\varpi
D^\dag(\alpha),$$ where $\alpha=(\alpha_1,\cdots, \alpha_N)^{\rm T} \in{\mathbb C^N}$, $D(\alpha)={\rm
	exp}[\sum_{j=1}^{N}(\alpha_j\hat{a}_j^\dag-\alpha_j^*\hat{a}_j)]$ is
the $N$-mode Weyl displacement operator and
$\varpi$  is a zero mean $N$-mode Gaussian state,
which is called the seed state of  $\Pi$ \cite{GC}.

For any bipartite Gaussian state, the authors in \cite{WJ}
derived a linear matrix inequality steering criterion via GPOVMs. Assume that $\rho\in{\mathcal S}(H_A\otimes H_B)$ is any $(m+n)$-mode Gaussian state with CM
$\Gamma_\rho$ in Eq.(\ref{N2}). As demonstrated in \cite{WJ},   $\rho$ is unsteerable (from A to B)by the subsystem
A's  all GPOVMs  if and only if
\begin{eqnarray}\label{N3}
	\Gamma_\rho+0_{2m}\oplus i\Omega_{n}\geq0;
\end{eqnarray}
and $\rho$ is unsteerable by the subsystem
B's all GPOVMs  if and only if
\begin{eqnarray*}
\Gamma_\rho+i\Omega_{m}\oplus0_{2n} \geq0.
\end{eqnarray*}
 Note that  any positive semidefnite matrix and its transpose matrix have the same eigenvalues. So
the condition \eqref{N3} is equivalent to  $$(\Gamma_\rho+0_{2m}\oplus i\Omega_{n})^{\rm T}\geq 0,$$
that is,
$$\Gamma_\rho-0_{2m}\oplus i\Omega_{n}\geq 0.$$
Hence
 $$\begin{array}{rl}&\Gamma_\rho+0_{2m}\oplus i\Omega_{n}\geq0\\
 	\Leftrightarrow&\Gamma_\rho \ge \pm(0_{2m}\oplus i\Omega_{n}) =\pm i(0_{2m}\oplus \Omega_{n}).\end{array}$$  

Denote  respectively by ${\mathcal {GS}}^{(N)}$ and ${\mathcal  {GS}^{(m,n)}_{\mathcal {US}(A \rightarrow  B)}}$ the set of all  $N$-mode Gaussian states,   the set of all $(m+n)$-mode  Gaussian states that are Gaussian
unsteerable (from A to B) by using Alice’s Gaussian measurements.  Particularly, if $N=m+n$, denote ${\mathcal {GS}}^{(N)}={\mathcal {GS}}^{(m,n)}. $ That is,
$${\mathcal {GS}}^{(N)}
	=\{\mbox{all N-mode Gaussian states}\}$$

and
$$\begin{array}{rl}&{\mathcal  {GS}^{(m,n)}_{\mathcal {US}(A \rightarrow  B)}}\\
	=&\{\rho \in{\mathcal {GS}}^{(m,n)} :\Gamma_{\rho}+0_{2m}\oplus i\Omega_{n}\geq0\}\\
	=&\{\rho \in {\mathcal {GS}}^{(m,n)} :\Gamma_{\rho}\geq \pm i(0_{2m}\oplus \Omega_{n})\}.\end{array}$$

 \subsection{ Gaussian  unsteerable  channels}

Recall that a Gaussian channel is a quantum channel which transforms any  Gaussian states into Gaussian states\cite{PMGH,CEGH}.  An $N$-mode Gaussian channel $\phi $ can be described by $\phi=\phi \left( {K,M,\mathbf d} \right)$, which acts on $\rho ( \Gamma_\rho ,\mathbf d_\rho )\in\mathcal {GS}(N)$ as
 \begin{eqnarray}\label{E11}
 	\mathbf d_\rho \mapsto K \mathbf d_\rho +\mathbf d,\ \
 	\Gamma_\rho  \mapsto K\Gamma_\rho K^{\rm T}+  M,
 \end{eqnarray}
where $\mathbf d \in \mathbb R^{2N}$ is a column displacement vector,  $K,M\in \mathcal M_{2N}(\mathbb R)$   satisfy $ M = M^{\rm T}$ and the completely positive condition
 \begin{eqnarray}\label{N0}
 	M + i\Omega_N  - iK\Omega_N {K^{\rm T}} \ge 0.
 \end{eqnarray}
Denote by 
${\mathcal {GC}}^{(N)}$ the set of all  $N$-mode Gaussian channels:
$${\mathcal {GC}}^{(N)}=\{  \mbox{all N-mode Gaussian quantum channels}\}. $$Particularly, if $N=m+n$, write ${\mathcal {GC}}^{(N)}={\mathcal {GC}}^{(m,n)}$.

Here, we give three known Gaussian channels which are used frequently.

{\it Attenuator channels.}  A single-mode attenuator channel $\phi_{\theta}^{n_{th}}(K,M,\mathbf d)$ is a deterministic Gaussian
	channel \cite{S1} with
\begin{eqnarray*}
	K  =
	\left({\begin{array}{*{20}{c}}\cos\theta&0\\ 0&\cos\theta\end{array}} \right) \ {\rm and}\
	M =
	\left( {\begin{array}{*{20}{c}}n_{th}\sin^2\theta &0\\0&n_{th}\sin^2\theta \end{array}} \right),
\end{eqnarray*}
where  $\theta \in [0,2\pi]$,  the thermal noise $n_{th}=2N_0+1\ge 1$ and  $N_0$ is mean number of thermal occupation. Particularly,  if $ n_{th}=1$, then the channel is  called  a pure lossy channel.

{\it Constant channels.}  An $N$-mode Gaussian channel  $\Theta(K,M,\mathbf d)$ is called a constant channel if  there exists some $N$-mode Gaussian state $\rho_0(\Gamma_0,  {\mathbf d}_0) \in \mathcal {GS}^{(N)}$ such that $\Theta(\rho)=\rho_0$  for all
$N$-mode Gaussian states 	$\rho\in  \mathcal{GS}^{(N)}$. In this case, $\Theta(K,M, \mathbf d)$ can be represented as $\Theta(K,M,\mathbf d)=\Theta(0, \Gamma_0,\mathbf {\mathbf d}_0)$.

{\it  Identity Gaussian channel.}  An $N$-mode Gaussian channel $\psi(K,M,\mathbf d) $ is called an identity channel if  $\psi(\rho )=\rho$ holds  for all $N$-mode Gaussian states
$\rho\in  \mathcal{GS}^{(N)}$. In this case, $K=I$,  $M=0$ and $\mathbf d=0$.

In \cite{YHQ}, the authors gave the definition of Gaussian unsteerable channels. Recall that any $(m+n)$-mode
Gaussian channel $\phi=\phi \left({K, M,\mathbf d} \right) $ is called Gaussian unsteerable from A to B if $K$ and $M$ satisfy the following relation
\begin{eqnarray}\label{N4}
	M + ({0_{2m}} \oplus i{\Omega _n}) -K\left( {0_{2m}} \oplus i{\Omega _n} \right){K^{\rm T}} \ge 0;
\end{eqnarray}
and  is called maximal Gaussian unsteerable from A to  B if  $\phi$ maps all $(m+n)$-mode Gaussian unsteerable states  from A to  B  into $(m+n)$-mode Gaussian unsteerable states  from A to  B.  

Let  $	{\mathcal {GC}}^{(m,n)}_{\mathcal {US}(A\rightarrow B)}$ and $	{\mathcal {GC}}^{(m,n)}_{\mathcal {MUS}(A\rightarrow B)}$ stand for the set of all  $(m+n)$-mode Gaussian unsteerable  (from A to  B)  channels   and the set of all $(m+n)$-mode maximal Gaussian unsteerable  (from A to  B)   channels, respectively, that is,
\[
\begin{aligned}
	{\mathcal {GC}}^{(m,n)}_{\mathcal {US}(A\rightarrow B)}
	=\{\phi(K,M,\mathbf d) \in {\mathcal {GC}}^{(m,n)}: \\
	M + (0_{2m} \oplus i{\Omega _n}) -K\left( 0_{2m} \oplus i{\Omega _n} \right){K^{\rm T}} \ge 0\},
\end{aligned}
\]
and
\[
	\begin{aligned}
{\mathcal {GC}}^{(m,n)}_{\mathcal {MUS}(A\rightarrow B)}
	=\{ \phi(K,M,\mathbf d) \in {\mathcal {GC}}^{(m,n)}: \\
	 \phi( \mathcal  {GS}^{(m,n)}_{\mathcal {US}(A \rightarrow  B)}) \subseteq \mathcal  {GS}^{(m,n)}_{\mathcal {US}(A \rightarrow  B)} \}.
\end{aligned}
\]

It is shown in \cite{YHQ} that  
the set  $	{\mathcal {GC}}^{(m,n)}_{\mathcal {US}(A\rightarrow B)} $is a proper  (but) large subset of  $	{\mathcal {GC}}^{(m,n)}_{\mathcal {MUS}(A\rightarrow B)}$.
 Obviously, the $(m+n)$-mode identity Gaussian channel   $\phi( I,0,0 ) \in 	{\mathcal {GC}}^{(m,n)}_{\mathcal {US}(A\rightarrow B)}$.

\subsection{Gaussian superchannels}

Recall that a superchannel is a completely positive linear map transforming
any quantum channels into quantum channels; and a Gaussian superchannel is
a superchannel transforming any Gaussian channels into Gaussian channels \cite{X11}. 

Denote by 
$${ \mathcal  {GSC}}^{(N)}=\{\mbox{all  N-mode Gaussian
	superchannels}   \}.$$
It is shown in \cite{X11} that any $N$-mode Gaussian superchannel $\Phi \in{\mathcal {GSC}}^{(N)}$ can be described by $\Phi \left( {A,E,Y, \nu }\right)$ as follows:  for any $N$-mode Gaussian channel $ \phi \left( {K,M,\mathbf d} \right) $,
$\Phi \left( {\phi \left( {K, M,\mathbf d}
	\right)} \right) = \phi '\left( {K', M',\mathbf d'} \right)$ with
\begin{eqnarray*}	K' = AK\Sigma_N E^{\rm T}\Sigma_N,\\
	M' = AMA^{\rm T} + Y, \\
	\mathbf d' = A\mathbf d +  \nu ,
\end{eqnarray*}
where $ A,E,Y \in \mathcal M_{2N}(\mathbb R)$, $Y = Y^{\rm T}$, $EE^{\rm T}=
{I_{2N}}$, $ \nu  \in {\mathbb R}^{2N}$,
$\Sigma_N=\underbrace{\Sigma\oplus\cdots\oplus\Sigma}_N$ with $\Sigma= {\begin{pmatrix}
		1&0\\
		0&{ - 1}
\end{pmatrix}}$, and
\begin{eqnarray}\label{1}
	Y + i\Omega_N  - iA\Omega_N A^{\rm T} \ge 0 , \nonumber\\
	i\Omega_N  - iE\Omega_N E^{\rm T} \ge 0.
\end{eqnarray}

 In addition, the author \cite{X11} also proved that any $N$-mode Gaussian superchannel $\Phi \left( {A,E,Y, \nu }\right)$ can be described in terms of compositions of Gaussian channels, that is,  for any $N$-mode Gaussian channel  $\phi \left( {K,M,\mathbf d} \right)$,  there exist two $N$-mode Gaussian channels ${\phi_1}\left( K_1, M_1,\mathbf d_1 \right)$ and  ${\phi _2}\left(	K_2, M_2, \mathbf d_2 \right) $ such that $\Phi ( \phi )= {\phi _2} \circ \phi  \circ {\phi _1}$.
	Particularly, one such representation is
		$$\begin{cases}K_1= \Sigma_N E^{\rm T} \Sigma_N, \ \ M_1 = 0,\ \ {\mathbf d_1} = 0;\\
			K_2 = A,\ \ M_2 = Y,\ \ {\mathbf d_2} = \upsilon.
		\end{cases}$$

\section{ Gaussian steering-annihilating   and  Gaussian steering-breaking channels }

In this section,  we will first give two  concepts of Gaussian steering-annihilating  channels and Gaussian   steering-breaking channels, and then discuss their properties.

	\begin{definition}\label{def4}
	Assume that  $\phi$ is  any   $(m+n)$-mode bipartite  Gaussian channel. We say that $\phi$   is    steering-annihilating  if it sends all $(m+n)$-mode Gaussian states    into either  $(m+n)$-mode Gaussian unsteerable states from A to B or $(m+n)$-mode Gaussian unsteerable states from B to A.
\end{definition}

	\begin{definition}\label{def41}
Assume that $\psi$ is any $m$-mode Gaussian channel. We say that  $\psi $  is  steering-breaking  if for any $(m+n)$-mode  Gaussian state $\rho$,   $(\psi \otimes I_K)( \rho)$ is always $(m+n)$-mode Gaussian unsteerable from A to B or from B to A,   where $K$ denotes arbitrary $n$-mode ancillary CV system.
\end{definition}

	  For the convinience,
denote by  $ {\mathcal {GC}}^{(m,n)}_{\mathcal {SA}}$,  ${\mathcal {GC}}^{(m,n)}_{\mathcal {SA} (A\rightarrow B)}$ the set of all  $(m+n)$-mode Gaussian steering-annihilating channels and 
the set of all $(m+n)$-mode  Gaussian steering-annihilating channels  sending all $(m+n)$-mode Gaussian states    into  $(m+n)$-mode Gaussian unsteerable states from A to B, respectively. Denote by $ {\mathcal {GC}}^{(N)}_{\mathcal {SB}}$,  $ {\mathcal {GC}}^{(N)}_{\mathcal {SB}(A\rightarrow B)}$ the set of all $N$-mode Gaussian steering-breaking channels and the set of all $N$-mode Gaussian steering-breaking channels from A to B, respectively.

 By Definitions \ref{def4} and \ref{def41},   the following useful property is obvious.

  \begin{proposition}  \label{P1}  If $\phi \in {\mathcal {GC}}^{(m,n)}_{\mathcal {SA}}$,  and   $\psi \in {\mathcal {GC}}^{(m,n)} $,  then $\phi \circ \psi \in {\mathcal {GC}}^{(m,n)}_{\mathcal {SA}}$; if $\phi \in {\mathcal {GC}}^{(N)}_{\mathcal {SB}}$ and $\psi \in {\mathcal {GC}}^{(N)} $,  then both $\phi \circ \psi$ and $\psi \circ \phi$ belong to  ${\mathcal {GC}}^{(N)}_{\mathcal {SB}}$.
  \end{proposition}

Next, we first give a sufficient condition for  Gaussian channels being Gaussian steering-annihilating. 

\begin{theorem}  \label{TH3}  Assume that $\phi =\phi\left( {K,M,\mathbf d} \right) \in {\mathcal {GC}}^{(m,n)}$ is any $(m+n)$-mode Gaussian channel. If $\phi$ satisfies one of the following conditions: 
		
 {\rm (1)}	 $M + ({0_{2m}} \oplus i{\Omega _n}) - iK\left( \Omega _m \oplus \Omega _n \right){K^{\rm T}} \ge 0,$
 
{\rm (2)} $M + ( i{\Omega _m} \oplus{0_{2n}}) - iK\left( \Omega _m \oplus \Omega _n \right){K^{\rm T}} \ge 0,$\\
then $\phi $ is steering-annihilating,  that is,
$\phi \in  {\mathcal {GC}}^{(m,n)}_{\mathcal {SA}}$.
\end{theorem}

\begin{proof}
	Assume that  $\rho  \in {\mathcal {GS}}^{(m,n)} $ is any Gaussian state  with CM $\Gamma_{\rho} $ in Eq.(\ref{N2}).  Then
	$\phi(\rho) $ has the CM $\Gamma_{\phi(\rho)} = K\Gamma_\rho {K^{\rm T}} + M$.
	 As
	$\Gamma_{\rho}  + i( \Omega_m \oplus \Omega_n ) \ge 0$,   we have  
	\begin{equation}\label{Th41}
			K(\Gamma_{\rho}+i( \Omega_m \oplus \Omega_n ))K^{\rm T} \ge 0 \end{equation} 
			for any $ K \in \mathcal M_{2(m+n)}(\mathbb R)$.  If  the condition (1) holds, combining  this assumption  with Ineq.\eqref{Th41},  one gets
\begin{eqnarray*}
& &\Gamma _{\phi ( \rho )} + ( 0_{2m} \oplus i{\Omega _n})\\
&=& K\Gamma_{\rho} {K^{\rm T}} + M + ( 0_{2m} \oplus i{\Omega _n} )\\
 &=& K(\Gamma_{\rho}+i( \Omega_m \oplus \Omega_n ))K^{\rm T}\\
& &+( M + ( 0_{2m} \oplus i\Omega_n)-K(i( \Omega_m \oplus \Omega_n ))K^{\rm T})\ge   0.
\end{eqnarray*}
It follows from  the steering criterion (Ineq.\eqref{N3})  that $\phi \left( \rho \right) \in \mathcal  {GS}^{(m,n)}_{\mathcal {US}(A \rightarrow  B)}$.

Similarly, one can show that, if the condition (2) holds, then  $\phi \left( \rho \right) \in \mathcal  {GS}^{(m,n)}_{\mathcal {US}(B \rightarrow  A)}$.
Hence $ \phi $ is Gaussian steering-annihilating.
\end{proof}

 Notice that the conditions in Theorem \ref{TH3} is only sufficient but not necessary for a channel being  steering-annihilating. In fact, there exist Gaussian steering-annihilating channels which do not satisfy this condition.

\begin{example}
	Take a $(1+1)$-mode Gaussian channel  $\phi_1=\phi(K_1,M_1,\mathbf{d_1})$, where	$M_1=I_4$ and
$${K_1} = \left( {\begin{array}{*{20}{c}}
 		{1.03}&0&0&0\\
 		0&{1.03}&0&0\\
 		0&0&{0.1}&0\\
 		0&0&0&{0.1}
 \end{array}} \right).$$
Then $\phi_1 $ is steering-annihilating from A to B, but does not satisfy the condition (1) in Theorem \ref{TH3}.
\end{example}

In fact, it is easily checked that $\phi \left( {{K_1},{M_1},{\rm d_1}} \right)$ satisfies the condition
		\[\begin{array}{rl}
	&	{M_1} + i\left( {{\Omega _1} \oplus {\Omega _1}} \right) - i{K_1}\left( {{\Omega _1} \oplus {\Omega _1}} \right)K_1^{\rm T}\\
		=& \left( {\begin{array}{*{20}{c}}
				1&{ - 0.0609i}&0&0\\
				{0.0609i}&1&0&0\\
				0&0&1&{0.99i}\\
				0&0&{ - 0.99i}&1
		\end{array}} \right)
		\ge 0,
	\end{array}\]
	but
	\[\begin{array}{rl}
	&	{M_1} + i\left( {{0_2} \oplus {\Omega _1}} \right) - i{K_1}\left( {{\Omega _1} \oplus {\Omega _1}} \right){K_1}^{\rm T}\\
		=& \left( {\begin{array}{*{20}{c}}
				1&{ - 1.0609i}&0&0\\
				{1.0609i}&1&0&0\\
				0&0&1&{0.99i}\\
				0&0&{ - 0.99i}&1
		\end{array}} \right)
		\ngeq 0.
	\end{array}\]
So $\phi$ does not satisfy the condition (1) in Theorem \ref{TH3}.

However, for any  $(1+1)$-mode Gaussian state $\rho$ with the standard CM    $\Gamma_\rho$ in Eq.\eqref{standard},
by a numerical calculation,	one can obtain
	\[\begin{array}{rl}
	&\Gamma_{\phi(\rho)}+\left( {{0_2} \oplus i{\Omega _1}} \right)=
	{K_1}\Gamma _\rho K_1^{\rm T} + {M_1} + \left( {{0_2} \oplus i{\Omega _1}} \right)\\
		=& {\small\left( {\begin{array}{*{20}{c}}
				{1.0609a+1}&0&{0.103c}&0\\
				0&{1.0609a + 1}&0&{0.103d}\\
				{0.103c}&0&{0.01b + 1}&i\\
				0&{0.103d}&{ - i}&{0.01b + 1}
		\end{array}} \right)}
		\ge 0,
	\end{array}\]
	which implies that $\phi$ is steering-annihilating (from A to B).

While  the condition (1) or (2)   in Theorem \ref{TH3} is not  a necessary condition, the subsequent  inequality in Theorem \ref{TH301} provides a sufficient and necessary condition for  a Gaussian channel becoming steering-annihilating.

\begin{theorem}  \label{TH301}   Assume that  $\phi =\phi\left( {K,M,\mathbf d} \right) \in {\mathcal {GC}}^{(m,n)}$ is any $(m+n)$-mode Gaussian channel. Then $\phi$ is Gaussian  steering-annihilating, that is, $\phi  \in  {\mathcal {GC}}^{(m,n)}_{\mathcal {SA}} $, if and only if 
\begin{equation}\label{N61}
\begin{aligned}
	&	\mathbf{w}^{\dag}M\mathbf{w} +|\mathbf{w}^{\dag}K\left( \Omega_{m} \oplus \Omega _n \right){K^{\rm T}}\mathbf{w}|\\
		\ge& |\mathbf{w}^{\dag}({0_{2m}} \oplus {\Omega _n})\mathbf{w}|
\end{aligned}	
\end{equation}
holds for  all $\mathbf{w} \in \mathbb{C}^{2(m+n)}$, 
 or 
\begin{equation}\label{N611}
		\begin{aligned}
			&	\mathbf{w}^{\dag}M\mathbf{w} +|\mathbf{w}^{\dag}K\left( \Omega_{m} \oplus \Omega _n \right){K^{\rm T}}\mathbf{w}|\\
			\ge& |\mathbf{w}^{\dag}( {\Omega _m} \oplus{0_{2n}})\mathbf{w}|
		\end{aligned}	
\end{equation}
holds for  all $\mathbf{w} \in \mathbb{C}^{2(m+n)}$. 
\end{theorem}

Take any  $N$-mode Gaussian channel
$\phi\in{\mathcal {GC}}^{(N)}$ with the state space $H$.
 Note that the pure state
 $$\left| {{\varphi _r}} \right\rangle  = \frac{1}{\cosh r}{\sum\limits_{j = 0}^\infty  {\left( {\tanh r} \right)} ^j}\left| j \right\rangle \left| j \right\rangle$$ is a two-mode Gaussian squeezed pure state, where $\{|j\rangle |j\rangle\}$ is  a tensor product of Fock states   and $r\in\mathbb R$ is the squeezed parameter.
 Correspondingly, $2N$-mode Gaussian squeezed pure state can be written as $\left| {{\psi _r}} \right\rangle=\left| {{\varphi _r}} \right\rangle^{\otimes N}.$
Define 
\begin{eqnarray*}\label{N111}
	{\rho_\phi }=  ( \phi  \otimes I_H)( \left| {{\psi _r}} \right\rangle \left\langle {{\psi _r}} \right|).
\end{eqnarray*}
Clearly, $\rho_\phi\in{\mathcal {GS}}^{(N,N)}$ is a $(N+N)$-mode Gaussian state. By the unsteerability of $\rho_\phi$, we can  give   necessary and sufficient conditions for Gaussian steering-breaking  channels.
	
\begin{theorem} \label{TH4}  Assume that  $\phi=\phi \left( {K,M,\mathbf d} \right) \in {\mathcal {GC}}^{(N)}$ is any $N$-mode  Gaussian channel.
	 Then the following statements are equivalent.
			
(1) $\phi$ is Gaussian steering-breaking, that is, $\phi   \in   {\mathcal{GC}}^{(N)}_{\mathcal {SB}}.$
			
(2) $ {\rho _\phi } \in {\mathcal  {GS}^{(N,N)}_{\mathcal {US}(A \rightarrow  B)}}$ or $ {\rho _\phi } \in {\mathcal  {GS}^{(N,N)}_{\mathcal {US}(B \rightarrow  A)}}$.
			
(3) The matrices  $ K$ and $M $ satisfy the condition
$$M-iK \Omega_N K^{\rm T} \ge 0 $$
or $$M+i\Omega_N \ge 0.$$
\end{theorem}
		
Proofs of Theorems 6 and 7 are given in Appendix.

In the end of this section, we will investigate the  relationship between   Gaussian steering-breaking  channels, Gaussian steering-annihilating  channels and
maximal   Gaussian unsteerable channels. Denote by 
$${\mathcal {GC}}^{(m,n)}_{\mathcal {MUS}}={\mathcal {GC}}^{(m,n)}_{\mathcal {MUS}(A \rightarrow B)} \bigcup {\mathcal {GC}}^{(m,n)}_{\mathcal {MUS}(B \rightarrow A)}. $$

By their definitions, it is obvious that

$$	{\mathcal {GC}}^{(m,n)}_{\mathcal {SA}}={\mathcal {GC}}^{(m,n)}_{\mathcal {SA}(A \rightarrow B)} \bigcup {\mathcal {GC}}^{(m,n)}_{\mathcal {SA}(B \rightarrow A)},$$
$$	{\mathcal {GC}}^{(m,n)}_{\mathcal {SB}}={\mathcal {GC}}^{(m,n)}_{\mathcal {SB}(A \rightarrow B)} \bigcup {\mathcal {GC}}^{(m,n)}_{\mathcal {SB}(B \rightarrow A)}$$
	 and
	$${\mathcal {GC}}^{(m,n)}_{\mathcal {SA}} \subset{\mathcal {GC}}^{(m,n)}_{\mathcal {MUS}}.$$
As a consequence
	of Proposition \ref{P1}, for any  $\phi\in {\mathcal {GC}}^{(m,n)}_{\mathcal {SA}} $  and  $\psi\in {\mathcal{GC}}^{(m+n)}_{\mathcal {SB}}$,  we have $$\phi \circ \psi \in{\mathcal {GC}}^{(m,n)}_{\mathcal {SA}} \cap {\mathcal{GC}}^{(m+n)}_{\mathcal {SB}}, $$
which means that  there are Gaussian channels which are
	simultaneously steering-breaking and steering-annihilating.
	
Next,  take an $(m+n)$-mode  constant channel $\Theta=\Theta(0,\Gamma_0,\mathbf d_0)$. Obviously,  by Theorem \ref{TH4},
$\Theta $  is steering-breaking.  However,  if  $\Gamma_0 +(0_{2m}\oplus i\Omega_{n})\ngeq 0 $,  and $\Gamma_0 +(i\Omega_{m}\oplus 0_{2n} )\ngeq 0$,  $\Theta$ is
neither steering-annihilating  nor maximal Gaussian unsteerable. So 
$$ {\mathcal{GC}}^{(m+n)}_{\mathcal {SB}}  \not\subset {\mathcal {GC}}^{(m,n)}_{\mathcal {SA} }$$ and
$$ {\mathcal{GC}}^{(m+n)}_{\mathcal {SB}}  \not\subset {\mathcal {GC}}^{(m,n)}_{\mathcal {MUS}}.$$

There also exist Gaussian steering-annihilating channels which are not Gaussian steering-breaking, that is,
$$ {\mathcal {GC}}^{(m,n)}_{\mathcal {SA}}  \not\subset  {\mathcal{GC}}^{(m+n)}_{\mathcal {SB}}.$$
For example, take  $\tilde{\phi}=\phi \otimes I_B \in {\mathcal {GC}}^{(1+1)}$,  where $\phi=\phi(\cos\theta I_2,\sin^2\theta I_2,\mathbf d_0)$ is a single-mode attenuator channel. Then $ \tilde{\phi} =\tilde{\phi}(K,M,\mathbf d) $ can be represented by
\begin{eqnarray*}
	K  =
	\left({\begin{array}{*{20}{c}}\cos\theta I_2&0_2\\ 0_2&I_2\end{array}} \right) {\rm and}\
		M =
		\left( {\begin{array}{*{20}{c}}\sin^2\theta I_2&0_2\\0_2&0_2\end{array}} \right).
\end{eqnarray*}
It can be easily checked that, when $0 \le \cos\theta \le \frac{\sqrt{2}}{2}$,
\begin{eqnarray*}		
&	M + ({0_2} \oplus i{\Omega _1}) - iK\left( \Omega _1 \oplus \Omega _1\right){K^T}\\
	=&\left(\begin{array}{cccc}\sin^2\theta &-i\cos^2\theta&0&0\\i\cos^2\theta&\sin^2\theta &0&0\\0&0&0&0\\
		0&0&0&0\end{array}\right)\ge 0,
\end{eqnarray*}	
 but
\begin{eqnarray*}
		&M - iK{\Omega _{2}}K^T
		 = M - iK\left( \Omega _1 \oplus \Omega_1 \right)K^T\\
		=&\left(\begin{array}{cccc}\sin^2\theta &-i\cos^2\theta&0&0\\i\cos^2\theta&\sin^2\theta&0&0\\0&0&0&-i\\
			0&0&i&0\end{array}\right)\ngeq 0.
\end{eqnarray*}	
 
 and \begin{eqnarray*}
 	&M + i{\Omega _{2}}
 	= M + i\left( \Omega _1 \oplus \Omega_1 \right)\\
 	=&\left(\begin{array}{cccc}\sin^2\theta &i&0&0\\-i&\sin^2\theta&0&0\\0&0&0&i\\
 		0&0&-i&0\end{array}\right)\ngeq 0.
 \end{eqnarray*}	
Thus, by  Theorem \ref{TH3} and Theorem \ref{TH4}, $\tilde{\phi} $ is steering-annihilating, but  not steering-breaking.

For the relation of these Gaussian channels,
see Fig.1.
  \begin{figure}[h]
 	\centering
 	\includegraphics[ width=0.5\textwidth]{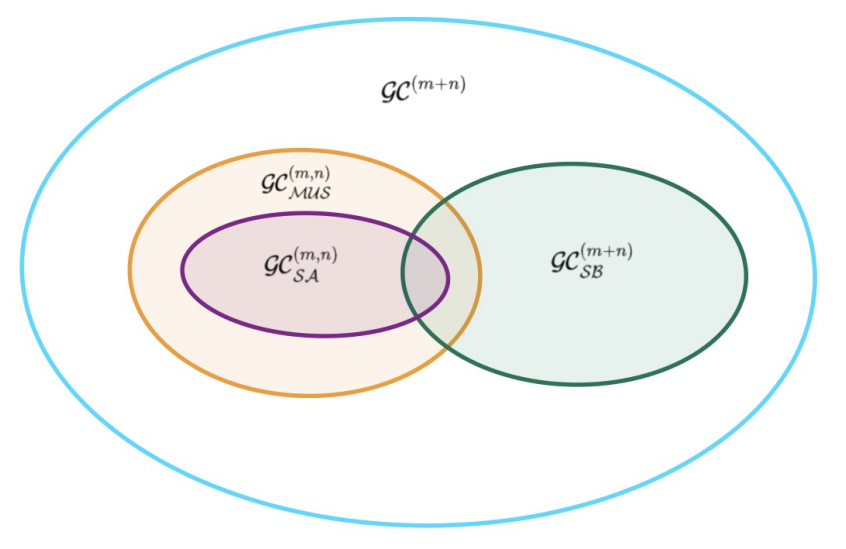}\\
 	 Fig. 1. The  relationship between   Gaussian steering-breaking  channels, Gaussian steering-annihilating  channels and
 	maximal  Gaussian unsteerable channels .
 \end{figure}

\section{Gaussian unsteerable superchannels}

In this section, we discuss two special type of Gaussian unsteerable  superchannels, that is, maximal Gaussian unsteerable superchannels and Gaussian unsteerable superchannels.
Here, we say that  {\it a Gaussian  superchannel  is  Gaussian unsteerable from A to B if it
maps any  Gaussian unsteerable (from A to B) channels into  Gaussian unsteerable  (from A to B)  channels; and is maximal Gaussian unsteerable  from A to B  if it
maps any maximal Gaussian unsteerable  (from A to B)  channels into maximal Gaussian unsteerable  (from A to B) channels. }

Denote by 
	${\mathcal {GSC}}^{(m,n)}_{\mathcal {MUS}(A \rightarrow B)}$
the set of all $(m+n)$-mode
maximal Gaussian unsteerable  (from A to B)  superchannels, that is,
$$\begin{array}{rl}&{\mathcal {GSC}}^{(m,n)}_{\mathcal {MUS}(A \rightarrow B)}\\
	=\{& \Phi \in {\mathcal {GSC}}^{(m,n)}: \\
	&\Phi({\mathcal {GC}}^{(m,n)}_{\mathcal {MUS}(A \rightarrow B)})\subseteq{\mathcal {GC}}^{(m,n)}_{\mathcal {MUS}(A \rightarrow B)} \},\end{array}$$
and by  ${\mathcal {GSC}}^{(m,n)}_{\mathcal {US}(A \rightarrow B)}$ the set of all  $(m+n)$-mode
 Gaussian unsteerable  (from A to B) superchannels, that is,
 $$\begin{array}{rl}&{\mathcal {GSC}}^{(m,n)}_{\mathcal {US}(A \rightarrow B)}\\
 	=\{& \Phi \in {\mathcal {GSC}}^{(m,n)}: \\
 	&\Phi({\mathcal {GC}}^{(m,n)}_{\mathcal {US}(A \rightarrow B)})\subseteq{\mathcal {GC}}^{(m,n)}_{\mathcal {US}(A \rightarrow B)} \},\end{array}$$

To characterize maximal Gaussian unsteerable superchannels, it is necessary to discuss the structure of
maximal Gaussian unsteerable channels.
By a similar argument to that of Theorem \ref{TH301},  we can obtain a necessary and sufficient condition for channels to be maximal unsteerable.

\begin{theorem}  \label{TH300}  Assume that  $\phi =\phi\left( {K,M,\mathbf d} \right) \in {\mathcal {GC}}^{\left(m,n \right)}$ is any $(m+n)$-mode Gaussian channel.
	Then   $\phi  \in	{\mathcal {GC}}^{(m,n)}_{\mathcal {MUS}(A \rightarrow B)}$ if and only if
\begin{equation} \label{N51}
\begin{aligned}
	&\mathbf{w}^{\dag}M\mathbf{w} +|\mathbf{w}^{\dag}K\left( 0_{2m} \oplus \Omega _n \right){K^{\rm T}}\mathbf{w}| \\
	\ge& |\mathbf{w}^{\dag}({0_{2m}} \oplus {\Omega _n})\mathbf{w}|
\end{aligned}	
\end{equation}
holds for  all $\mathbf{w} \in \mathbb{C}^{2(m+n)}$.
\end{theorem}

Thus, by Theorem
\ref{TH300},  we can give sufficient conditions of Gaussian superchannels becoming maximal unsteerable and unsteerable, respectively.

\begin{theorem}  \label{TH501}
	Assume that $\Phi=\Phi \left( A,E,Y,\nu \right) \in  {\mathcal {GSC}}^{(m,n)} $ is  any $(m+n)$-mode  Gaussian superchannel.
If either 

(1) there exist maximal Gaussian unsteerabe channels $\chi_1=\chi_1( K_1, M_1,\mathbf d_1 )$, $\chi_2=\chi_2(K_2, M_2, \mathbf d_2 ) \in  {\mathcal {GC}}^{(m,n)}_{\mathcal {MUS}(A \rightarrow B)}$
	such that  $\Phi(\phi)= \chi_2 \circ \phi \circ \chi_1$ for all  $ \phi \left( {K,M,\mathbf d} \right) \in
\mathcal {GC}^{(m,n)}$, 
	or
	
	(2)	the matrices $A,E$ and $Y$ satisfy
	$$\begin{array}{rl}
		&\mathbf{w}^{\dag}Y\mathbf{w}+|\mathbf{w}^{\dag} A\left( 0_{2m}\oplus \Omega _n \right)A^{\rm T}\mathbf{w}| \\
		\ge&|\mathbf{w}^{\dag}( 0_{2m} \oplus \Omega _n )\mathbf{w}|
	\end{array}$$
	and
	$$\begin{array}{rl}
		&|\mathbf{w}^{\dag}\Sigma_{m+n}E^{\rm T} \Sigma_{m+n} \left( 0_{2m} \oplus \Omega _n\right)\Sigma_{m+n} E \Sigma_{m+n}\mathbf{w}| \\
		\ge&|\mathbf{w}^{\dag}( 0_{2m} \oplus \Omega_n ) \mathbf{w}|
	\end{array}	$$
	for all $\mathbf{w} \in \mathbb{C}^{2(m+n)}$,
then  $\Phi\in {\mathcal {GSC}}^{(m,n)}_{\mathcal {MUS}(A \rightarrow B)}$.
	
	If either
	
	(3)  there exist Gaussian unsteerabe channels $\chi_1=\chi_1( K_1, M_1,\mathbf d_1 )$, $\chi_2=\chi_2(K_2, M_2, \mathbf d_2 )\in {\mathcal {GC}}^{(m,n)}_{\mathcal {US}(A \rightarrow B)}$  such that  $\Phi(\phi)= \chi_2 \circ \phi \circ \chi_1$ for all  $ \phi \left( {K,M,\mathbf d} \right) \in \mathcal {GC}^{(m,n)}$, 
	or
	
	(4)	the matrices $A,E$ and $Y$ satisfy $$	Y + \left( 0_{2m} \oplus i\Omega _n \right) - A\left(0_{2m} \oplus i\Omega _n \right){A^{\rm T}} \ge 0$$	and
$$
		\left( 0_{2m} \oplus i\Omega_n \right) -
		E \left( 0_{2m} \oplus i\Omega _n\right)E^{\rm T} \ge 0,$$ then 
$\Phi\in {\mathcal {GSC}}^{(m,n)}_{\mathcal {US}(A \rightarrow B)}$.
\end{theorem}

Proofs of Theorems 8 and 9 are given in Appendix.

Notice that the converse of Theorem \ref{TH501}  may not be true. In fact, if ${\chi_2}$ is a steering-annihilating Gaussian channel, by Proposition \ref{P1}, 
$\Phi \left( A,E,Y,\nu \right)$ always is a maximal Gaussian  unsteerable superchannel, regardless of the property of $\chi_1$.

In addition, there exists maximal Gaussian  unsteerable superchannels which are not  Gaussian  unsteerable superchannels.
For example,
take a $(1+1)$-mode  Gaussian superchannel $ \Phi=\Phi \left( A,E,Y,\nu \right) $ with
	\[A = \left( {\begin{array}{*{20}{c}}
			{0.170929}&{ - 0.942009}&{ - 0.609808}&{ - 0.108889}\\
			{1.301268}&{0.599464}&{0.666952}&{ - 0.800351}\\
			{ - 0.151061}&{ - 0.241749}&{0.938864}&{1.130728}\\
			{0.441668}&{1.125889}&{ - 1.767416}&{0.418528}
	\end{array}} \right),\]
\[Y = \left( {\begin{array}{*{20}{c}}
		{5.890063}&{ - 1.845370}&{2.502275}&{ - 1.763982}\\
		{ - 1.845370}&{5.297160}&{ - 2.573896}&{ - 2.759869}\\
		{2.502275}&{ - 2.573896}&{4.270381}&{0.944184}\\
		{ - 1.763982}&{ - 2.759869}&{0.944184}&{3.732230}
\end{array}} \right)\]
and		\[E = \left( {\begin{array}{*{20}{c}}
		1&0&0&0\\
		0&1&0&0\\
		0&0&1&0\\
		0&0&0&1
\end{array}} \right).\]
By calculations, we find that $	A,E,Y $ satisfy
	\[{{w^\dag }Yw + \left| {{w^\dag }A\left( {{0_2} \oplus {\Omega _1}} \right){A^{\rm T}}w} \right| \ge \left| {{w^\dag }\left( {{0_2} \oplus {\Omega _1}} \right)w} \right|}\]
	for	any $w \in \mathbb{C}^4$,
but  do not satisfy
\[{Y - i\left( {{0_2} \oplus {\Omega _1}} \right) + A\left( {{0_2} \oplus i{\Omega_1}} \right){A^{\rm T}} \ge 0 .}\]
That implies that  $\Phi \left( A,E,Y,\nu \right) \in {\mathcal {GSC}}^{(1,1)}_{\mathcal {MUS}(A \rightarrow B)}$, but $\Phi \left( A,E,Y,\nu \right) \notin {\mathcal {GSC}}^{(1,1)}_{\mathcal {US}(A \rightarrow B)}$.
Thus, the set of Gaussian maximal unsteerable superchannels differs from that of all Gaussian unsteerable superchannels. Nevertheless, both can be identified as free superchannels within the resource theory of steering for channels in CV systems.

\section{Conclusion}

As a core component of CV systems, the research on Gaussian channels holds indispensable theoretical value and practical significance. Gaussian channels not only serve as a typical and tractable research vehicle for quantum resource theory, acting as an ideal model to analyze core quantum resources such as coherence, entanglement, and EPR steering, but also deepen the understanding of fundamental physical issues including quantum system symmetry and environmental decoherence, thereby improving the axiomatic framework and mathematical formulation of quantum information theory.
As one of quantum resources,  EPR steering is a unique quantum resource situated between quantum entanglement and Bell nonlocality, which  play  significant roles in various quantum protocols, secure communication and other fields.

In this work, we investigate several classes of Gaussian channels associated with EPR steering: Gaussian steering-annihilating channels (which completely eliminate steering), Gaussian steering-breaking channels (which locally disrupt steering), and maximal Gaussian unsteerable channels. We derive the necessary and sufficient conditions for a Gaussian channel to belong to each of these classes, respectively, and clarify the relationships among them. Notably, there exist channels that are simultaneously steering-annihilating and steering-breaking.
In addition, from the perspective of quantum resource theory, we discuss the structure of Gaussian unsteerable superchannels as free operations, and establish the necessary and sufficient conditions for Gaussian channels to be maximal unsteerable. These results provide deeper insights into quantum channels in CV systems. Future work may focus on constructing a resource theory for quantifying the steering capability of Gaussian channels in CV systems, as well as exploring other quantum resources—such as entanglement—for bosonic Gaussian channels and non-Gaussian channels.\\

\section{Acknowledments }   The authors wish to give their thanks to the referees for many
helpful comments to improve the original paper.

This work is partially supported by National Natural Science Foundation of China (11901421, 12571138, 12171290), Shanxi Scholarship Council of China (2025-001) and the special fund for Science and Technology Innovation Teams of Shanxi	Province (202304051001035).

\section*{Appendix}

\begin{proof}[A proof of Theorem 6]
	Assume that $\phi =\phi\left( {K,M,\mathbf d} \right) \in {\mathcal {GC}}^{(m,n)}_{\mathcal {SA}}$ is any $(m+n)$-mode Gaussian steering-annihilating channel.  Then for any $\rho  \in {\mathcal {GS}}^{(m,n)} $ with CM $\Gamma_{\rho} $,
	$\phi(\rho)$ is unsteerable from A to B or from B to A.
	
	 If $\phi(\rho)\in {\mathcal {GS}}^{(m,n)}_{\mathcal {US}(A \rightarrow B)}$,  by  Eqs.\eqref{N3}-\eqref{E11},  one has 
	\begin{eqnarray*}\label{N62}
		K\Gamma_{\rho}K^{\rm T}+M  \ge  \pm i(0_{2m} \oplus \Omega_n),
	\end{eqnarray*}
	and so
		\begin{eqnarray*}
			\mathbf{w}^{\dag}(K\Gamma_{\rho}K^{\rm T}+M)\mathbf{w}  \ge \pm i\mathbf{w}^{\dag} (0_{2m} \oplus \Omega _n)\mathbf{w},
		\end{eqnarray*}
		that is,
	\begin{eqnarray*}
		\mathbf{w}^{\dag}K\Gamma_{\rho}K^{\rm T}\mathbf{w}+\mathbf{w}^{\dag}M\mathbf{w}  \ge \pm i\mathbf{w}^{\dag} (0_{2m} \oplus \Omega _n)\mathbf{w}
	\end{eqnarray*}
	holds  for all $ \mathbf{w} \in \mathbb{C}^{2(m+n)}$.
	
	Note that, it is shown in \cite{PMGH} that,  for any $N$-mode Gaussian state $\rho$ with  CM $\Gamma$, it holds that
	\begin{eqnarray}\label{inf1}	\inf\limits_{\Gamma\geq\pm i\Omega_N}{\bf w}^\dag\Gamma{\bf w}=|{\bf w}^\dag\Omega_N{\bf w}|,\ \ \forall  \mathbf{w} \in \mathbb{C}^{2N}.\end{eqnarray}
	
	Thus, by Eq.\eqref{inf1}, we have
	$$	\begin{array}{rl}
		&\inf\limits_{\Gamma_\rho\geq\pm i(\Omega_{m} \oplus \Omega _n)}	\mathbf{w}^{\dag}K\Gamma_{\rho}K^{\rm T}\mathbf{w}+\mathbf{w}^{\dag}M\mathbf{w} \\
		=&|\mathbf{w}^{\dag}K(\Omega_{m} \oplus \Omega _n)K^{\rm T}\mathbf{w}|+\mathbf{w}^{\dag}M\mathbf{w} \\
		\ge& \pm i\mathbf{w}^{\dag} (0_{2m} \oplus \Omega _n)\mathbf{w},
	\end{array}$$
	which implies
	$$|\mathbf{w}^{\dag}K(\Omega_{m} \oplus \Omega _n)K^{\rm T}\mathbf{w}|+\mathbf{w}^{\dag}M\mathbf{w}  \ge |\mathbf{w}^{\dag} (0_{2m} \oplus \Omega _n)\mathbf{w}|$$ for all $\mathbf{w} \in \mathbb{C}^{2(m+n)},$
	that is, Ineq.\eqref{N61} holds.
	
	Symmetrically,  one can show that, if $\phi(\rho)\in{\mathcal {GS}}^{(m,n)}_{\mathcal {US}(B \rightarrow A)}$,  then Ineq.\eqref{N611} holds.
	
	On the other hand, if Ineq.\eqref{N61} holds, then for any $ \mathbf{w} \in \mathbb{C}^{2(m+n)}$, by using Eq.\eqref{inf1}, one has
	\begin{eqnarray*}
		\mathbf{w}^{\dag}\Gamma _{\phi ( \rho )}\mathbf{w}
		&=&\mathbf{w}^{\dag}K\Gamma_{\rho} {K^{\rm T}}\mathbf{w} +\mathbf{w}^{\dag} M \mathbf{w}\\
		&\ge& \inf_{\Gamma_{\rho} \ge  \pm i\left( {\Omega_m \oplus {\Omega _n}} \right) } \mathbf{w}^{\dag}K\Gamma_{\rho} {K^{\rm T}}\mathbf{w} +\mathbf{w}^{\dag} M \mathbf{w}\\
		&=& |\mathbf{w}^{\dag}K(\Omega_m \oplus \Omega _n)K^{\rm T}\mathbf{w}|+\mathbf{w}^{\dag}M\mathbf{w}  \\
		&\ge& |\mathbf{w}^{\dag} (0_{2m} \oplus \Omega _n)\mathbf{w}|\\
		&\geq& \pm i\mathbf{w}^{\dag} (0_{2m} \oplus \Omega _n)\mathbf{w} .
	\end{eqnarray*}
	So $\Gamma _{\phi ( \rho )}\ge \pm i(0_{2m} \oplus \Omega _n)$, which means that  $\phi(\rho)$ is unsteerable from A to B. 
Symmetrically, if Ineq.\eqref{N611} holds, then $\phi(\rho)$ is unsteerable from B to A.
	It follows  that $\phi \in {\mathcal {GC}}^{(m,n)}_{\mathcal {SA}}$.
\end{proof}

\begin{proof}[A proof of Theorem 7]	Assume that $\phi=\phi \left( {K,M,\mathbf d} \right) \in {\mathcal {GC}}^{\left(N\right)}$ is any $N$-mode Gaussian channel.
	
	$\left( 1 \right) \Rightarrow \left( 2 \right)$:  By Definition \ref{def41}, this is obvious.
	
	$\left( 2 \right) \Rightarrow \left( 3 \right)$:
Note that  the CM $\Gamma _{{\rho _\phi }}$    of ${\rho _\phi }$ has the form \cite{X11}
	\begin{eqnarray*}\label{Choi}
		{\Gamma _{{\rho _\phi }}} = \left( {\begin{array}{*{20}{c}}
				{\cosh2r K{K^{\rm T}} + M}&{\sinh2r K\Sigma_{N}  }\\
				{\sinh2r\Sigma_{N} {{K^{\rm T}}} }&{\cosh2rI_{2N}}
		\end{array}} \right).
	\end{eqnarray*}		
 If  $ {\rho _\phi } \in {\mathcal {GS}}^{(N,N)}_{\mathcal {US}(A \rightarrow B)}$, then 
	\begin{equation}\label{NN}
		\Gamma _{{\rho _\phi }} + ( 0_{2N} \oplus i{\Omega_{N} }) \ge 0.
	\end{equation}			
	It is well known \cite{ZH} that  a Hermitian matrix  
	\begin{eqnarray}\label{ZH}
	&	W=\left(\begin{array}{cc}W_{11}& W_{12}\\
			W_{12}^{\dag} & W_{22}\end{array}\right)\ge 0 \\ \nonumber
			 &\Leftrightarrow W_{22}\ge 0,\ \ W_{11}-W_{12}W_{22}^{-1}W_{12}^{\dag}\ge0\\\nonumber
		&\Leftrightarrow W_{11}\ge 0,\ \ W_{22}-W_{12}^{\dag}W_{11}^{-1}W_{12}\ge0.
		\end{eqnarray}
		
	So Ineq.\eqref{NN} implies
	$$\begin{cases}\label{66}
		\cosh2rI_{2N} + i{\Omega _{N}}\ge 0,\\	
		\begin{array}{rl}
			&\cosh2r K{K^{\rm T}} + M \\
			&- \sinh^2(2r) K\Sigma_N (\cosh2r I_{2N} + i\Omega _{N} )^{-1}\Sigma_N {K^{\rm T}} \ge 0.\end{array}\end{cases}$$
	Note that  $(T+cI)^{-1}=\frac{1}{c}(I-\frac{T}{c})+o(c^{-2})$ \cite{S1}, where  $o(c^{-2})$ stands for the Landau little-$o$  which will be neglected when taking to  $c \to \infty $. So
	\begin{widetext}
		\begin{eqnarray*}
			&&0\leq \mathop {\lim }\limits_{r \to \infty } [\cosh2r K{K^{\rm T}}  + M - \sinh^2(2r) K\Sigma_N ( \cosh2r I_{2N} + i{\Omega _N} )^{ - 1}\Sigma_N K^{\rm T}]\\
			&=& \mathop {\lim }\limits_{r \to \infty } [\cosh2rK{K^{\rm T}}  + M - \sinh^2(2r) K\Sigma_N[\frac{1}{\cosh2r}( I_{2N} -\frac{i\Omega _N}{\cosh2r})+o(\frac{1}{\cosh2r})]\Sigma_N{K^{\rm T}} ]\\
			&=& \mathop {\lim }\limits_{r \to \infty } [\cosh2r K{K^{\rm T}} + M - K\Sigma_N[\frac{\sinh^2(2r)}{\cosh2r}I_{2N} - i\tanh^2(2r){\Omega _N} +\sinh^2(2r)o(\frac{1}{\cosh2r})]\Sigma_N K^{\rm T} ]= M -i K{\Omega _N}{K^{\rm T}}.
		\end{eqnarray*}	
	\end{widetext}
 If  $ {\rho _\phi } \in {\mathcal {GS}}^{(N,N)}_{\mathcal {US}(B \rightarrow A)}$,
then
 \begin{equation}\label{NN8}
 	\Gamma _{{\rho _\phi }} + (i\Omega_{N}  \oplus 0_{2N} ) \ge 0.
 \end{equation}	
By  the equivalence \eqref{ZH} again, Ineq.\eqref{NN8} implies
$$\begin{array}{rl}
	&\cosh2r K{K^{\rm T}} + M+ i\Omega _{N}\\
	- &\sinh^2(2r) K\Sigma_N (\cosh2r I_{2N})^{-1}\Sigma_N {K^{\rm T}} \ge 0.\end{array}$$
So
	\begin{eqnarray*}
	&&0\leq \mathop {\lim }\limits_{r \to \infty } [\cosh2r K{K^{\rm T}} + M+ i\Omega _{N}\\
	&-&\sinh^2(2r) K\Sigma_N (\cosh2r I_{2N})^{-1}\Sigma_N {K^{\rm T}}]\\
	&=& \mathop {\lim }\limits_{r \to \infty } [\frac{\cosh^2(2r) -\sinh^2(2r)}{\cosh2r}K{K^{\rm T} }+ M + i\Omega _{N}] \\
	&=& M +i{\Omega _N}.
\end{eqnarray*}	
Therefore, the statement  (2) holds.

	$\left( 3 \right) \Rightarrow \left( 1 \right)$:
	Assume that $H'$ is  any  separable complex Hilbert space with the responding $N'$-mode CV system.
	Take any  $\rho  \in  \mathcal {GS}^{\left( N, N'\right)}$ with  CM
	$
	{\Gamma _\rho } = \left( {\begin{array}{*{20}{c}}
			X&Z\\
			{{Z^{\rm T}}}&Y
	\end{array}} \right)$.
	Then
	the CM of $\left( {\phi  \otimes  I_{H'}} \right)(\rho)$ has the form
		\begin{eqnarray}\label{one}
	\begin{array}{rl}
		&{\Gamma _{\left( {\phi  \otimes  I_{H'}} \right)(\rho )}}\\
		=& \left( {\begin{array}{*{20}{c}}
				K&{0}\\
				{0}&I
		\end{array}} \right)\left( {\begin{array}{*{20}{c}}
				X&Z\\
				{{Z^{\rm T}}}&Y
		\end{array}} \right)\left( {\begin{array}{*{20}{c}}
				{{K^{\rm T}}}&{0}\\
				{0}&I
		\end{array}} \right) + \left( {\begin{array}{*{20}{c}}
				M&{0}\\
				{0}&0
		\end{array}} \right)\\
		= &\left( {\begin{array}{*{20}{c}}
				{KX{K^{\rm T}} + M}&{KZ}\\
				{{Z^{\rm T}}{K^{\rm T}}}&Y
		\end{array}} \right).
	\end{array}
		\end{eqnarray}
	Since    $\Gamma _\rho+i({\Omega _N} \oplus {\Omega _{N'}})\ge 0$, a direct calculation gives
	$${Y + i{\Omega _{N'}} \ge 0}$$and
	$$	{X + i{\Omega _N} - Z\left( {Y + i{\Omega _{N'}}} \right)^{-1}{Z^{\rm T}} \ge 0}.$$

	If the matrices  $ K$ and $M $ satisfy  $M-iK \Omega_N K^{\rm T} \ge 0$, one gets
		\begin{eqnarray*}\label{666}
			&&	KX{K^{\rm T}} + M - KZ{\left( {Y + i{\Omega _{N'}}} \right)^{ - 1}}{Z^{\rm T}}{K^{\rm T}}\\
			&=& K\left( {X - Z{{\left( {Y + i{\Omega _{N'}}} \right)}^{ - 1}}{Z^{\rm T}}} \right){K^{\rm T}} + M\\
			&=& K\left({X - Z{{\left( {Y + i{\Omega _{N'}}} \right)}^{ - 1}}{Z^{\rm T}}}
			+ i{\Omega _N} \right){K^{\rm T}} \\
			&&	+ (M - iK{\Omega _N}{K^{\rm T}} )\ge0.
		\end{eqnarray*}
	
	Also note that $KX{K^{\rm T}} + M \ge 0
	$
	by the denifition of the channel $\phi$.
	By \cite{ZH} again,  one achieves 
	$${\Gamma _{\left( {\phi  \otimes  I_{H'}} \right)(\rho )}} + i(0_N \oplus {\Omega _{N'}})\ge 0.$$ So  $ (\phi  \otimes { I_{H'}})\left( \rho  \right)$
	is unsteerable from A to B,  and hence  $\phi$	is steering-breaking.
	
 Now, assume that the matrices  $ K$ and $M $ satisfy the condition $M+i\Omega_N \ge 0$.  
It is shown that in \cite{LHQ} that, if the matrices $ L,V,W \in \mathcal M_n(\mathbb C)$ satisfy the condition  $L> 0$, $W\geq 0$ and   $VLV^{\dag}+W$  invertible, then  $$V^{\dag}(VLV^{\dag}+W)^{-1}V \le L^{-1}.$$ For $X$, $K$ and $M$, applying the above result, one achieves 
 $${K^{\rm T}}(KX{K^{\rm T}} + (M+i{\Omega _{N}}))^{-1}K\le X^{-1}.$$  and so
		\begin{eqnarray*}
		&& Y-{Z^{\rm T}}{K^{\rm T}}(KX{K^{\rm T}} + M+i{\Omega _{N}})^{-1}KZ \\
		&\ge & Y-{Z^{\rm T}}X^{ - 1}{Z} \ge0.
	\end{eqnarray*}

By the equivalence \eqref{ZH},  one obtains
$${\Gamma _{\left( {\phi  \otimes  I_{H'}} \right)(\rho )}} + i({\Omega _{N}} \oplus 0_N')\ge 0.$$ So  $ (\phi  \otimes { I_{H'}})\left( \rho  \right)$
is unsteerable from B to A. Hence  $\phi$	is steering-breaking.
\end{proof}

	\begin{proof}[A proof of Theorem 8]
		Assume that  $\phi =\phi\left( {K,M,\mathbf d} \right) \in {\mathcal {GC}}^{\left(m,n \right)}$ is any $(m+n)$-mode Gaussian channel and $\rho\in {\mathcal {GS}}^{(m,n)}_{\mathcal {US}(A \rightarrow B)}$ is any $(m+n)$-mode Gaussian unsteerable state with CM $\Gamma_{\rho}$.
		
		On the one hand, if  $\phi \in  {\mathcal {GC}}^{(m,n)}_{\mathcal {MUS}(A \rightarrow B)}$, 
		then 
		$\phi(\rho)\in {\mathcal {GS}}^{(m,n)}_{\mathcal {US}(A \rightarrow B)}$. So
		\begin{eqnarray*}
			K\Gamma_{\rho}K^{\rm T}+M  \ge \pm i( 0_{2m} \oplus \Omega _n ),
		\end{eqnarray*}
		which implies  
		\begin{eqnarray*}\label{N53}
			\mathbf{w}^{\dag}K\Gamma_{\rho}K^{\rm T}\mathbf{w}+\mathbf{w}^{\dag}M\mathbf{w}  \ge \pm i \mathbf{w}^{\dag} (0_{2m} \oplus \Omega _n)\mathbf{w}
		\end{eqnarray*}
		for all	$\mathbf{w} \in \mathbb{C}^{2(m+n)}$.
		By a similar discussion to that of Eq.\eqref{inf1},  one can obtain
		$$\begin{array}{rl}
			&\inf\limits_{\Gamma_{\rho} \ge  \pm i( 0_{2m} \oplus \Omega _n)}	\mathbf{w}^{\dag}K\Gamma_{\rho}K^{\rm T}\mathbf{w}+\mathbf{w}^{\dag}M\mathbf{w}  \\
			=& |{\bf w}^\dag K(0_{2m} \oplus \Omega _n)K^{\rm T}{\bf w}|+\mathbf{w}^{\dag}M\mathbf{w}  \\
			\ge &\pm i\mathbf{w}^{\dag} (0_{2m} \oplus \Omega _n)\mathbf{w}.
		\end{array}$$
		Hence
		\begin{eqnarray*}
			|\mathbf{w}^{\dag}K(0_{2m} \oplus \Omega _n)K^{\rm T}\mathbf{w}|+\mathbf{w}^{\dag}M\mathbf{w} \ge |\mathbf{w}^{\dag} (0_{2m} \oplus \Omega _n)\mathbf{w}|
		\end{eqnarray*}
		holds for all 	 $\mathbf{w} \in \mathbb{C}^{2(m+n)}$.
		
		On the other hand,  if \eqref{N51} holds, as $\Gamma_{\rho} \ge  \pm i\left( { 0_{2m} \oplus {\Omega _n}} \right) $,    for any $ \mathbf{w} \in \mathbb{C}^{2(m+n)}$, 
		we have
		\begin{eqnarray*}
			\mathbf{w}^{\dag}\Gamma _{\phi ( \rho )}\mathbf{w}
			&=&\mathbf{w}^{\dag}K\Gamma_{\rho} {K^{\rm T}}\mathbf{w} +\mathbf{w}^{\dag} M \mathbf{w}\\
			&\ge& \inf_{\Gamma_{\rho} \ge \pm i( 0_{2m} \oplus \Omega _n )} \mathbf{w}^{\dag}K\Gamma_{\rho} {K^{\rm T}}\mathbf{w} +\mathbf{w}^{\dag} M \mathbf{w}\\
			&=& |\mathbf{w}^{\dag}K(0_{2m} \oplus \Omega _n)K^{\rm T}\mathbf{w}|+\mathbf{w}^{\dag}M\mathbf{w}  \\
			&\ge& |\mathbf{w}^{\dag} (0_{2m} \oplus \Omega _n)\mathbf{w}|\\
			&\ge& \pm i\mathbf{w}^{\dag } (0_{2m} \oplus \Omega _n)\mathbf{w} .
		\end{eqnarray*}
		This means  $\Gamma _{\phi ( \rho )}\ge \pm i(0_{2m} \oplus \Omega _n)$, and so $\phi(\rho)$ is unsteerable from A to B.
		It  follows that  $\phi \in {\mathcal {GC}}^{(m,n)}_{\mathcal {MUS}(A \rightarrow B)}$.
	\end{proof}
	
    \begin{proof}[A proof of Theorem 9]

Suppose  that   $\Phi \left( A,E,Y,\nu \right) \in  {\mathcal {SGC}}^{(m,n)} $ is  any  $(m+n)$-mode Gaussian superchannel.
By \cite{X11},   there exist some  $\chi_1=\chi_1( K_1, M_1,\mathbf d_1 )$, $\chi_2=\chi_2(K_2, M_2, \mathbf d_2 ) \in {\mathcal {GC}}^{(m,n)}$   such that
$\Phi \left( \phi  \right) = {\chi_2} \circ \phi  \circ \chi_1=\phi'(K',M',\mathbf d)$ for all  $ \phi \left( K,M,\mathbf d \right) \in \mathcal {GC}^{(m,n)}$.  Note that  $K'$ and $M'$ satisfy  
$$M' ={K_2}K{M_1}{K^{\rm T}}{K_2}^{\rm T} + {K_2}M{K_2}^{\rm T} + {M_2}$$ and 
$$K'=K_2{K}K_1. $$

Also
notice that  the CM $\Gamma_\tau$ of any $(m+n)$-mode unsteerable (from A to B) Gaussian state $\tau$  can be written as $$\Gamma_{\tau} =   0_{2m} \oplus Q_{\tau}  +P_{\tau}$$
with some real matrix $Q_{\tau}\ge i{\Omega _n}$ and $ P_{\tau}\geq 0$ \cite{MKIB}.

Now, if both $\chi_1$ and $\chi_2$ are maximal unsteerable from A to B,  then for any $(m+n)$-mode maximal Gaussian unsteerable channel   $ \phi=\phi \left( K,M,\mathbf d \right) \in {\mathcal {GC}}^{(m,n)}_{\mathcal {MUS}(A \rightarrow B)} $ and any Gaussian unsteerable state $\rho=\rho(\Gamma_\rho,\mathbf d_\rho)\in {\mathcal {GS}}^{(m,n)}_{\mathcal {US}(A \rightarrow B)}$,
we have 
\begin{eqnarray*}
	&&\Gamma_{\Phi \left( \phi  \right)(\rho)}=\Gamma _{{\chi_2} \circ \phi  \circ {\chi_1}( \rho )}\\
	&=&{K}_2K{K_1}\Gamma_\rho{K_1}^{\rm T}{K^{\rm T}}{K_2}^{\rm T}\\
	&&+{K_2}K{M_1}{K^{\rm T}}{K_2}^{\rm T} + {K_2}M{K_2}^{\rm T} + {M_2}\\
	&=& {K}_2[K( 0_{2m} \oplus Q_{\rho}^{(1)} + P_{\rho}^{(1)} )K^{\rm T}+M]{K}_2^{\rm T}+M_2\\
	&=& {K}_2( 0_{2m} \oplus Q_{\rho}^{(2)} +P_{\rho}^{(2)}){K}_2^{\rm T}+M_2\\
	&\ge& \pm i(0_{2m} \oplus {\Omega _n}),
\end{eqnarray*}
where
$$0_{2m} \oplus Q_{\rho}^{(1)} + P_{\rho}^{(1)}=K_1\Gamma_{\rho} K_1^{\rm T}+M_1,$$
$$ 0_{2m} \oplus Q_{\rho}^{(2)} +P_{\rho}^{(2)}=K\left(  0_{2m} \oplus Q_{\rho}^{(1)} + P_{\rho}^{(1)} \right)K^{\rm T}+M ,$$
and	the second equation, the third equation,  the last inequality due to respectively the maximal unsteerability of $\chi_1$, $\phi$  and $\chi_2$.  So $\Phi(\phi)(\rho)$ is unsteerable from A to B. 
It follows  that  $\Phi \left( \phi  \right) \in {\mathcal {GC}}^{(m,n)}_{\mathcal {MUS}(A \rightarrow B)} $.

If both $\chi_1$ and  $\chi_2$ are Gaussian unsteerable from A to B, that is, they satisfy Ineq.\eqref{N4}:
\begin{eqnarray*}\label{N7}
	M_2 + \left( 0_{2m} \oplus i\Omega _n \right) - K_2\left(0_{2m} \oplus i\Omega _n \right){K_2^{\rm T}} \ge 0
\end{eqnarray*}	and
\begin{eqnarray*}\label{N72}
	M_1 + \left( 0_{2m} \oplus i\Omega _n \right) - K_1\left(0_{2m} \oplus i\Omega _n \right){K_1^{\rm T}} \ge 0,
\end{eqnarray*}	
then for any  $(m+n)$-mode Gaussian unsteerable channel $\phi( {K, M,\mathbf d})\in {\mathcal {GC}}^{(m,n)}_{\mathcal {US}(A \rightarrow B)} $, we have 
\begin{eqnarray*}\label{N72}
	M+ \left( 0_{2m} \oplus i\Omega _n \right) - K\left(0_{2m} \oplus i\Omega _n \right){K^{\rm T}} \ge 0,
\end{eqnarray*}	and so 

\begin{eqnarray*}\label{N9}
	& &M' + ({0_{2m}} \oplus i{\Omega _n}) -K'\left( {0_{2m}} \oplus i{\Omega _n} \right){K'^{\rm T}} \\
	&=& {K_2}K{M_1}{K^{\rm T}}{K_2}^{\rm T} + {K_2}M{K_2}^{\rm T} + {M_2}\\
	&& +({0_{2m}} \oplus i{\Omega _n})-K_2{K}K_1 (0_{2m} \oplus i{\Omega _n})K_1^{\rm T}K^{\rm T} {{ K_2}^{\rm T}}	\\
	&=&  {K_2}K[M_1+({0_{2m}} \oplus i{\Omega _n})-K_1 (0_{2m} \oplus i{\Omega _n})K_1^{\rm T}]K^{\rm T}{K_2}^{\rm T}\\
	&& + {K_2}[M+({0_{2m}} \oplus i{\Omega _n})-K (0_{2m} \oplus i{\Omega _n})K^{\rm T}]{K_2}^{\rm T} \\
	&&+{M_2}+({0_{2m}} \oplus i{\Omega _n})-K_2 (0_{2m} \oplus i{\Omega _n}) {{ K_2}^{\rm T}}\\
	&\ge&  0.
\end{eqnarray*}
Hence  $ \Phi \left( \phi  \right)\in {\mathcal {GC}}^{(m,n)}_{\mathcal {US}(A \rightarrow B)} $.

On the other hand, suppose that the matrices $A,E$ and $Y$ satisfy
\begin{eqnarray*}\label{N7}
	Y + \left( 0_{2m} \oplus i\Omega _n \right) - A\left(0_{2m} \oplus i\Omega _n \right){A^{\rm T}} \ge 0
\end{eqnarray*}	and
\begin{eqnarray*}\label{N72}
	\left( 0_{2m} \oplus i\Omega_n \right) -
	E \left( 0_{2m} \oplus i\Omega _n\right)E^{\rm T} \ge 0. 
\end{eqnarray*}	
As $\Sigma_{m+n} \left( 0_{2m} \oplus i\Omega _n\right)
\Sigma_{m+n}=-\left( 0_{2m} \oplus i\Omega _n\right)$, the inequality implies
$$\begin{array}{rl}
	&	\left( 0_{2m} \oplus i\Omega_n \right) \\
	&- \Sigma_{m+n}
	E^{\rm T}\Sigma_{m+n} \left( 0_{2m} \oplus i\Omega _n\right)
	\Sigma_{m+n} E \Sigma_{m+n}  \ge 0.
\end{array}	$$
Thus, for any $\phi(K,M,\mathbf d)\in {\mathcal {GC}}^{(m,n)}_{\mathcal {US}(A \rightarrow B)} $ with $\Phi(\phi)=\phi'(K',M',\mathbf d')$, 
$K'$ and $M'$ satisfy
\begin{eqnarray*}\label{N9}
	& &M' + ({0_{2m}} \oplus i{\Omega _n}) -K'\left( {0_{2m}} \oplus i{\Omega _n} \right){K'^{\rm T}} \\
	&=& {A}MA^{\rm T}+Y+(0_{2m} \oplus i\Omega _n)\\
	&&	-{A}{K}\Sigma_{m+n} E^{\rm T}\Sigma_{m+n} (0_{2m} \oplus i{\Omega _n})\Sigma_{m+n} E\Sigma_{m+n} {{K^{\rm T}}} {{ A}^{\rm T}}	\\
	&\ge& AMA^{\rm T}+Y- AK(0_{2m} \oplus i\Omega_n) K^{\rm T} A^{\rm T}+(0_{2m} \oplus i\Omega _n) \\
	&=& A[M+(0_{2m} \oplus i\Omega _n)-K(0_{2m} \oplus i\Omega_n) K^{\rm T}] A^{\rm T}\\
	&&	+ Y+(0_{2m} \oplus i\Omega _n)-A(0_{2m} \oplus i\Omega _n)A^{\rm T}\\
	&\ge& Y+(0_{2m} \oplus i\Omega _n)-A(0_{2m} \oplus i\Omega _n)A^{\rm T} \ge 0.
\end{eqnarray*}
Hence   $ \Phi \left( \phi  \right)\in {\mathcal {GC}}^{(m,n)}_{\mathcal {US}(A \rightarrow B)} $ in this case.
\end{proof}

\end{document}